\documentclass[sigplan,10pt]{acmart}
\acmConference[PL'18]{ACM SIGPLAN Conference on Programming Languages}{January 01--03, 2018}{New York, NY, USA}
\acmYear{2018}
\acmISBN{} 
\acmDOI{} 
\startPage{1}

\setcopyright{none}

\usepackage{booktabs}   
\usepackage{subcaption} 

\usepackage{xspace}
\usepackage{graphicx}
\usepackage{pdfpages}
\usepackage{amsmath}
\usepackage{amsthm}
\usepackage{amscd}
\usepackage{amssymb}
\usepackage{float}
\usepackage{epsfig}
\usepackage{xcolor}
\usepackage{color}
\usepackage{subcaption}
\usepackage{filecontents}
\usepackage{moresize}
\usepackage{stmaryrd}
\usepackage{mathrsfs}
\usepackage{mathtools}

\usepackage{multicol}

\usepackage{mathpartir}

\usepackage{hhline}
\usepackage{changepage}
\usepackage{wrapfig}

\usepackage{xr}
\usepackage{threeparttable}
\usepackage{epstopdf}

\newcommand{\defeq}{\stackrel{\textrm{def}}{=}}

\newcommand{\ite}{\texttt{ite}}

\let\vec\boldsymbol

\newcommand{\domain}{\textit{\textbf{G}}}

\newcommand{\mso}{\textsf{MSO}}
\newcommand{\smso}{\textsf{sMSO}}
\newcommand{\fo}{\textsf{FO}}
\newcommand{\sampling}{\textrm{sampling}}
\newcommand{\filtering}{\textrm{filtering}}
\newcommand{\summ}{\textrm{sum}}
\newcommand{\compare}{\textsf{compare}}
\newcommand{\update}{\textsf{update}}
\newcommand{\outputs}{\textsf{output}}
\newcommand{\SRT}{\textsf{SRT}}
\newcommand{\DSRT}{\textsf{DSRT}}

\newcommand{\addfree}{\textsf{SRT}$_{\textsf{A}}$}
\newcommand{\daddfree}{\textsf{DSRT}$_{\textsf{A}}$}
\newcommand{\uninit}{\textsf{SRT}$_{\textsf{U}}$}
\newcommand{\dense}{\textsf{SRT}$_{\textsf{D}}$}
\newcommand{\AU}{\textsf{SRT}$_{\textsf{AU}}$}
\newcommand{\AD}{\textsf{SRT}$_{\textsf{AD}}$}
\newcommand{\AUD}{\textsf{SRT}$_{\textsf{AUD}}$}

\newcommand{\cg}{\mathcal{CG}}

\newcommand{\oldv}{\textsf{old}}
\newcommand{\newv}{\textsf{new}}
\newcommand{\addv}{\textsf{add}}
\newcommand{\pos}{\textsf{pos}}
\newcommand{\negz}{\textsf{neg}}

\newcommand{\hide}[1]{}

\newcommand{\etal}{\textit{et al.}\@\xspace}

\numberwithin{equation}{section}

\definecolor{dkgreen}{rgb}{0,0.3,0}
\definecolor{gray}{rgb}{0.5,0.5,0.5}
\definecolor{mauve}{rgb}{0.58,0,0.82}
\definecolor{light-gray}{gray}{0.80}

\usepackage{listings}
\lstdefinelanguage{spmd}{
  morekeywords = {
       loc, bit, bool, true, false
     , implements, harness
     , null
     , assert, assume
     , else
     , find, fix, fold, for, forall, function
     , generator, gen
     , if, while, int, float, bool, string
     , loop, simple, cond, val
     , fork, join
     , nil, null, none, new, malloc
     , option, or
     , ref, return
     , spmdfork, nprocs, spmdtransfer
     , void
     , concrete, sym
     , requires, ensures
     , invariant, decreases
     , conj, exp
     , init, stmt},
  literate=
    {-}{--}1,
  morecomment=[l]{//}
}
\renewcommand{\scriptsize}{\fontsize{8.5}{9}\selectfont}
\makeatletter
\lst@AddToHook{TextStyle}{\let\lst@basicstyle\scriptsize\fontfamily{lmss}\selectfont}
\makeatother

\usepackage{url}

\usepackage{rotating}

\usepackage{enumitem}

\usepackage[linesnumbered]{algorithm2e}
\usepackage{pstricks}
\usepackage{multirow}

\newcommand{\conf}[1]{}

\newcommand{\Strand}{{\sf Strand}\xspace}

\newcommand{\semantics}[1]{\llbracket{}#1\rrbracket{}}

\usepackage{booktabs}   
\usepackage{subcaption} 

\begin{document}

\title{Streaming Transformations of Infinite Ordered-Data Words}         


\author{Xiaokang Qiu}
\orcid{nnnn-nnnn-nnnn-nnnn}             
\affiliation{
  \department{School of Electrical and Computer Engineering}              
  \institution{Purdue University}            
  \streetaddress{Street1 Address1}
  \city{West Lafayette}
  \state{IN}
  \postcode{47907}
  \country{USA}                    
}
\email{xkqiu@purdue.edu}          


\begin{abstract}
In this paper, we define streaming register transducer (\SRT), a one-way, letter-to-letter, transductional machine model for transformations of infinite data words whose data domain forms a linear group. Comparing with existing data word transducers, \SRT~ are able to perform two extra operations on the registers: a linear-order-based comparison and an additive update. We consider the transformations that can be defined by \SRT~ and several subclasses of \SRT. We investigate the expressiveness of these languages and several decision problems. Our main results include: 1) \SRT~ are closed under union and intersection, and add-free \SRT~ are also closed under composition; 2) \SRT-definable transformations can be defined in monadic second-order (\mso) logic, but are not comparable with first-order (\fo) definable transformations; 3) the functionality problem is decidable for add-free \SRT, the reactivity problem and inclusion problem are decidable for deterministic add-free \SRT, but none of these problems is decidable in general for \SRT.
\end{abstract}



\maketitle

\section{Introduction}
\label{sec:intro}

Transformations of infinite strings describe the behavior of many computing systems, especially reactive systems. Several logic- and automata-based models~\cite{Gire1986,Varricchio1993,Beal2004,Alur2012} have been proposed to describe these transformations, with various expressiveness. Nonetheless, all these models only handle letters from a finite alphabet. Strings with a data value from an infinite domain for each position, usually called \emph{data words}, are needed in various application scenarios. For example, data words naturally describe paths in XML trees or linked-lists data structures. 

Transformations of data words pose a set of new challenges to the formalism as they need to support operations that compare and manipulate data values from an infinite domain. Classical finite-state automata for finite alphabet are extended to models that recognize finite data words by checking, at least the equality of data values. These models include register automata~\cite{ra94,Demri_2009}, data automata~\cite{bojanczyk,Bojanczyk2010} and pebble automata~\cite{pebble,Neven_2004}.
More recently, automata models are also proposed to process data words with a linearly-ordered data domain~\cite{Segoufin2011,Tan2012}.
Several decidable logics are also proposed to describe properties of data words and decision procedures are developed~\cite{havoc,csl,popl11}.

What remains less explored, however, is a natural and powerful machine model for implementations of transformations of data words. To be close to real-world systems, the model should preferably have several desirable features. 
We describe them below.

First, the model should be a one-way, letter-to-letter transducer, similar to a Mealy machine. This allows us to model streaming transformations of infinite data words. which have become increasingly common today. For example, in the Internet of Things (IoT): most things generate, process and transmit streaming data continuously. In other words, the system maintains only a fix amount of data values, and just sequentially visits each piece of {\bf non-persistent data} once.

Second, the model needs to be expressive enough to describe common operations over streamed data, e.g., sampling, filtering, aggregation, etc. Sampling calls for {\bf nondeterminism} of the model (a piece of data can be randomly sampled or not sampled); filtering calls for comparisons of data values in a {\bf linear order} (e.g., to describe, low-pass or high-pass filtering); aggregation calls for an {\bf additive operation} over the data domain (e.g., to compute the sum of a group of data values).

Thirdly, the model should be general and {\bf agnostic to the underlying data domain}. In other words, study results for this model can be applied to transformations over arbitrary data domain, e.g., integers, rationals, or reals.


While researchers have proposed several transductional models in recent years, they hardly meet the desirable features stated above. Streaming data-string transducers (SDST)~\cite{Alur2011} are deterministic transducers that transform finite data words only. SDST feature a set of data string variables, which allow the model to memoize unbounded number of data values, making the model not suitable for non-persistent data manipulation. Moreover, the transducer supports linear-order comparison of data values only. The models proposed in~\cite{Durand2016} are similar but do not even allow data comparison.
Similarly, the model for reactive system implementation proposed by Ehlers~\etal~\cite{Ehlers_2014} is too powerful as it allows unbounded memory through a queue. Meanwhile it is also too weak as it supports equality check only.
Register transducers~\cite{Khalimov_2018,Exibard2019} are studied as implementations of register automata as specification. Similarly, these models are deterministic and support equality check only.

To this end, we propose \emph{streaming register transducer}~(\SRT), a model we argue naturally describes streaming transformations of infinite ordered-data words, whose underlying data domain forms a {\bf linear group}. Examples of linear groups include integers, rationals or reals together with $\leq$ and $+$, and multi-dimensional planes of these numbers with lexicographic order and point-wise addition. \SRT~ accept a finite or infinite data word as input, performs one-way, letter-to-letter transformations and produces another finite or infinite data word as output. Similar to register transducer~\cite{Khalimov_2018,Exibard2019}, \SRT~ are equipped with finite states and a fixed number of registers, and the transitions are determined by the current state and the data comparison between registers and the current data value. However, \SRT~ support nondeterministic transitions and rich operations for linear group: linear-order-based data comparison and updates to registers by adding the current data value to them. We also investigate several subclasses of \SRT: the transitions are deterministic, the additive updates are disallowed, the registers are uninitialized, or the data domain is dense.

\begin{table*}
\begin{tabular}{|l|c|c|c|c|c|c|c|c|}
\hline
 & $\cup$-closed & $\cap$-closed & Comp-closed & \fo-definable & \mso-definable & Functionality & Reactivity \& Inclusion \\ \hline
 \SRT & yes & yes & no & no & yes & undecidable & undecidable \\ \hline
 \DSRT & no & yes & no & no & yes & trivial & undecidable \\ \hline
  \addfree & yes & yes & yes & no & yes & {\sf 2NEXPTIME} & open \\ \hline
 \daddfree & no & yes & yes & no & yes & trivial & {\sf 2NEXPTIME} \\ \hline
  \AD & yes & yes & yes & no & yes & {\sf NEXPTIME} & open \\ \hline  
  \AU & yes & yes & yes & no & yes & {\sf NEXPTIME} & open \\ \hline  
\end{tabular}
\caption{Summary of main results.}
\label{tbl:results}
\end{table*}

The main results of the paper are summarized in Table~\ref{tbl:results}.
First, \SRT~ are closed under union and intersection, and add-free. \SRT~ are also closed under composition.
Second, \SRT-definable transformations can be defined in monadic second-order (\mso) logic, but are not comparable with first-order (\fo) definable transformations. More precise logical characterization of \SRT~ is posed as an open problem.
Thirdly, the functionality problem is decidable for add-free \SRT, the reactivity problem and inclusion problem are decidable for deterministic add-free \SRT, but none of these problems is decidable in general for \SRT. The reactivity and inclusion problems for nondeterministic \SRT~ remain open.

\section{Preliminaries}

\subsection{Data Words}
\label{sec:dw}

\begin{definition}[Linear Group]
A \emph{linear group} is a triple $(D, \leq, +)$ where $D$ is infinite data domain, $(\leq)$ is a total order over $D$, and $+$ is a binary operation such that $(D, +)$ forms an additive group, i.e., $+$ is associative, has an identity $0$ and inverse operation $-$.
\end{definition}

The most common instances of linear group are $(\mathbb{Z}, \leq, +)$, $(\mathbb{Q}, \leq, +)$ and $(\mathbb{R}, \leq, +)$, where $\mathbb{Z}$, $\mathbb{Q}$ and $\mathbb{R}$ are the sets of integers, rationals and reals, respectively. These primary linear groups can also be combined to form a multi-dimensional plane. For example, $(\mathbb{R}^2, \leq, +)$ is a linear group, where $\mathbb{R}^2$ is the bidimensional Euclidean space of reals, where $\leq$ is the lexicographic order (i.e., $(a, b) \leq (c, d)$ if $a < c$, or $a=c$ and $b \leq d$), $+$ is the point-wise addition (i.e., $(a, b) + (c, d)$ is defined as $(a+c, b+d)$).

\begin{definition}[Density]
A linear order $(D, \leq)$ is \emph{dense} if for any two elements $a, b \in D$ such that $a < b$, there exists another $c \in D$ such that $a < c < b$.
\end{definition}

\begin{definition}[Discreteness]
A linear group $(D, \leq, +)$ is \emph{discrete} if there exists a least positive element.
\end{definition}

\begin{proposition}
\label{thm:dense-discrete}
A linear group is discrete if and only if it is not dense.
\end{proposition}
In other words, there is a dichotomy of dense and discrete linear groups. For example, $(\mathbb{Z}, \leq, +)$ is discrete; $(\mathbb{Q}, \leq, +)$ and $(\mathbb{R}, \leq, +)$ are dense. In the rest of the paper, we will interchangeably use term non-dense or discrete.


Let $\Sigma$ and $\Gamma$ be finite sets of labels, and let $(D, \leq)$ be a linear group. A finite data word (resp, $\omega$-data word) over an alphabet $\Sigma \times D$ is a finite sequence (resp. $\omega$-sequence) of letters from $\Sigma \times D$. We write $(\Sigma \times D)^*$ (resp. $(\Sigma \times D)^*$) for the set of finite data words (resp. $\omega$-data words), and $(\Sigma \times D)^{\infty}$ for $(\Sigma \times D)^* \cup (\Sigma \times D)^*$. 

For a data word $s$, we write $s[i]$ for the $i$-th letter of $s$, and $|s|$ for the length of $s$. And for any letter $a$, we represent the set of positions in $s$ with letter $a$ as \[ \langle s \rangle_a: \{p \in \mathbb{N} \mid s[p] = a \}\]
We also write $s[P]$ to represent the string contracted from $s$ based on the position set $P$, i.e., $s[P] \defeq \langle s[i] \rangle_{i \in P}$.

Given, two data words $s \in (\Sigma \times D)^{\infty}$ and $t \in (\Gamma \times D)^{\infty}$, if $|s| = |t|$, they can be combined to form a transformation instance:
\[ s \otimes t \defeq \langle (s[i], t[i]) \rangle_{i \geq 0}\]

\subsection{Streaming Transformation}

\begin{definition}[Streaming Transformation]
\label{def:transformation}
Let $\Sigma$ be an input label set, $\Sigma$ be an output label set, $D$ be a data domain. A \emph{$(\Sigma, \Gamma, D)$-streaming transformation} (or \emph{transformation} for short) 
is a language over the input-output pairs $(\Sigma \times D) \times (\Gamma \times D)$ such that if $V$ is a word in the language, so is every prefix of $V$. We call a word in a transformation a $(\Sigma, \Gamma, D)$-\emph{transformation instance}.
\end{definition}
Intuitively, streaming transformations are mappings from input to output that can be done in a letter-to-letter fashion.
Note that the definition assumes every input letter triggers an output letter. To allow the no-output behavior for some input letters, one can simply introduce $\#$ as a special output label and use $(\#, 0)$ as a vacuous output.
A transformation instance $V$ can be uniquely split $V$ into an input data word and an output data word, denoted as $in(V)$ and $out(V)$.

\begin{definition}[Transformation Composition]
Given a $(\Sigma, \Gamma, D)$-transformation $\mathcal{T}_1$ and another $(\Gamma, \Theta, D)$-transformation $\mathcal{T}_2$, the composition transformation $\mathcal{T}_1 \cdot \mathcal{T}_2$ is defined as a $(\Sigma, \Theta, D)$-transformation that contains exactly instances of the form $s_1 \otimes s_2$ that satisfies the following conditions:
\begin{itemize}
\item $s_1 \in (\Sigma \times D)^*$ and $s_2 \in (\Theta \times D)^*$;
\item there exists a data word $s_3 \in (\Gamma \times D)^*$ such that $s_1 \otimes s_3 \in \mathcal{T}_1$ and $s_3 \otimes s_2 \in \mathcal{T}_2$.
\end{itemize}
\end{definition}

\subsection{Monadic Second-Order Logic}

If $\Sigma$ is a finite set of labels and $\domain = (D, \leq, +)$ is a linear group, we define $\mso(\Sigma, \domain)$ formulae that can be interpreted on finite data words. A $\mso(\Sigma, \domain)$ formula is built up from atomic formulae of the following forms:
\[
x=y \quad x \in X \quad L_{\sigma}(x) \quad S(x, y) \quad E \leq 0
\]
where $x, y$ are first-order position variables ranging over positions of the data word, $X$ is a second-order variable ranging over sets of positions, and $=$ and $\leq$ are interpreted in the natural way. $\sigma \in \Sigma$ is a label and $L_{\sigma}(x)$ states that the letter at the $x$-th position is labeled with $\sigma$. $S(x, y)$ states that $y$ is the position next to $x$, i.e., $y = x+1$. $E$ is an expression evaluated to values in $D$ and $E \leq 0$ just states that the value of $E$ is not positive. The expression consists of terms of the following forms:
\[
dt(x) \quad sum\_dt(X) \quad 0
\]
where $dt(x)$ is the data value stored in the $x$-th position, and $sum\_dt(X)$ is the sum of all data values stored in the positions represented by $X$. These terms can be combined arbitrarily using the inverse operation $-$ and the additive operation $+$. 

Atomic formulae are connected with boolean connectives $\neg, \land, \lor$ and (first-order or second-order) quantifiers $\exists, \forall$.


The logic can be extended to $\mso(\Sigma, \Gamma, \domain)$, being interpreted on $(\Sigma, \Gamma, D)$-transformation instances.
 Besides those allowed in $\mso(\Sigma, \domain)$, there is one extra form of atomic formulae: 
\[
L_{\gamma}(x)
\]
which states that the output for the $x$-th position is labeled $\gamma$. In addition, an expression can refer to the data values of both the input and the output with the following terms:
\[
dt_{in}(x) \quad dt_{out}(x) \quad sum\_dt_{in}(X) \quad sum\_dt_{out}(X) \quad 0
\]

If an $\mso(\Sigma, \Gamma, \domain)$ formula contains first-order variables and quantifiers only, i.e., it is set-variable-free, we call it a first-order formula.
The set of first-order formulae is denoted as $\fo(\Sigma, \Gamma, \domain)$. 

\section{Streaming Register Transducer}

In this section, we introduce streaming register transducer, the model we propose to recognize streaming transformations. We present its definition and several subclasses of it, and discuss closure properties.

\begin{definition}[\SRT]
\label{def:srt}
A \emph{streaming register transducer} (\SRT) is defined as an octuple $(\Sigma, \Gamma, \domain, Q, q_0, k, R_0, \Delta)$ where
\begin{itemize}
\item $\Sigma$ and $\Gamma$ are the finite sets of input and output labels, respectively;
\item $\domain = (D, \leq, +)$ is a linear group;
\item $Q$ is a finite set of states and $q_0 \in Q$ is the initial state;
\item $k \in \mathbb{N}$ is the number of registers;
\item $R_0 \in D^{k}$ is the initial assignments to the registers;
\item $\Delta \subseteq Q \times \Sigma \times \{>, =, <\}^{k} \times \{\oldv, \newv, \addv\}^{k} \times \{1, \dots, k\} \times \Gamma \times Q$ is a set of (nondeterministic) transitions.
\end{itemize}
\end{definition}

Intuitively, an \SRT~ takes as input an $\omega$-word over the alphabet $\Sigma \times D$ and performs letter-to-letter transformation and outputs a finite or $\omega$-word over the alphabet $\Gamma \times D$. It maintains a finite state and a set of registers that store values from $D$. We also call it a $(\Sigma, \Gamma, \domain)$-\SRT~ (or just \SRT~ if the signature is arbitrary) and define it as $(Q, q_0, k, R_0, \Delta)$. 
We use $\SRT(k)$ to represent the class of \SRT~ with $k$ registers.

Whenever a letter $(\sigma, d)$ is read from the head, it checks the label $\sigma$ and compares the data value $d$ with the current values stores in the registers with respect to the linear order, and makes a state transition. In addition, it also updates each register with one of four possible values: the current value, the input value $d$, the sum of the two values, or $0$. It then optionally outputs a letter in the output alphabet $\Gamma \times D$, whose data value is copied from one of the registers. 
We next formally define the configurations, runs and trails of \SRT.

\subsection{Configuration and Run}
\label{sec:config}

A \emph{configuration} of an \SRT~ as defined in Definition~\ref{def:srt} is of the form $(q, R)$ where $q \in Q$ is the current state and $R \in \domain^{k}$ is the current values stored in the $k$ registers.
A \emph{transition step} between two configurations can be represented as $(q, R) \xrightarrow[{(\gamma, d')}]{(\sigma,d)} (q', R')$: from an old configuration $(q, R)$, the transducer reads the next letter $(\sigma, d)$ from the input stream, transits to a new configuration $(q', R')$, and appends $(\gamma, d')$ to the output stream. 

The transition step is enabled by a transition $(q, \sigma, l, m, u, \gamma, q') \in \Delta$ if the following conditions are satisfied:
\begin{enumerate}
\item there is a vector $l \in \{>, =, <\}^{k}$ such that for any $1 \leq i \leq k$, one of the following three conditions holds:
\begin{itemize}
\item $d > R[i]$ and $l[i]$ is $>$; 
\item $d = R[i]$ and $l[i]$ is $=$; 
\item $d < R[i]$ and $l[i]$ is $<$.
\end{itemize}
\item for any $1 \leq m \leq k$, one of the following conditions holds:
\begin{itemize}
\item $m[i] = \oldv$ and $R'[i] = R[i]$;
\item $m[i] = \newv$ and $R'[i] = d$; or
\item $m[i] = \addv$ and $R'[i] = R[i] + d$;
\end{itemize}
\item 
$t = (\gamma, R'[u])$.
\end{enumerate}

\begin{definition}[Run]
Let $\mathcal{S} = (Q, q_0, k, R_0, \Delta)$ be a $(\Sigma, \Gamma, D)$-\SRT~ and let $s \in (\Sigma \times \domain)^*$ be a finite data word. A \emph{run} $\rho$ of $\mathcal{S}$ over $s$ that generates $t$ is a sequence of $n$ transition steps
\[ \rho : ~(q_0, R_0) \xrightarrow[{t[0]}]{s[0]} (q_1, R_1) \xrightarrow[{t[1]}]{s[1]} \cdots \xrightarrow[{t[n-1]}]{s[n-1]} (q_n, R_n) \]
such that for every $0 \leq i < n$, the transition step $(q_i, R_i) \xrightarrow[{t[i]}]{s[i]} (q_{i+1}, R_{i+1})$ is enabled by a transition in $\Delta$.
\end{definition}
Notice that for any $0 < m < n$, the subsequence of $\rho$ from $(q_0, R_0)$ to $(q_m, R_m)$ is also a run over the partial input $s[0..m-1]$. We call a it a \emph{prefix} of $\rho$ and denote it as $\rho[0..m]$.

The run $\rho$ actually defines a transformation instance, i.e., an infinite sequence of input-output pairs: 
\[ \semantics{\rho} \defeq s \otimes \langle t_i \rangle_{i \geq 0}\]
Any \SRT~ defines a streaming transformation:
\[\semantics{\mathcal{S}} \defeq \{\semantics{\rho} \mid \rho \textrm{ is a run of } \mathcal{T} \} \]
We call a transformation \emph{\SRT-definable} if there is an \SRT~ that defines the transformation.

To illustrate the expressiveness of \SRT, we now show three examples that do aggregation, sampling and filtering, respectively. 

\begin{figure}
\begin{subfigure}[b]{\columnwidth}
\begin{center}
\unitlength=2.5mm
\includegraphics{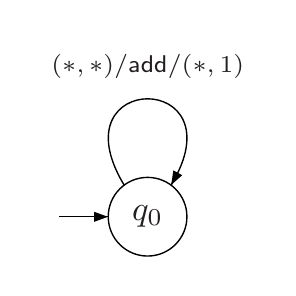}
\end{center}
\caption{$\mathcal{S}_{\summ}$ in Example~\ref{ex:sum}} \label{fig:sum}
\end{subfigure}

\begin{subfigure}[b]{\columnwidth}
\begin{center}
\unitlength=2.5mm
\includegraphics{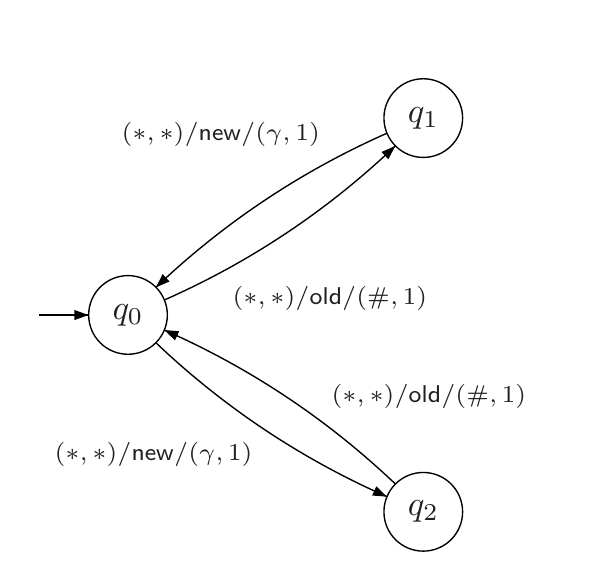}
\end{center}
\caption{$\mathcal{S}^2_{\sampling}$ in Example~\ref{ex:sampling}} \label{fig:sampling}
\end{subfigure}

\begin{subfigure}[b]{\columnwidth}
\begin{center}
\unitlength=2.5mm
\includegraphics{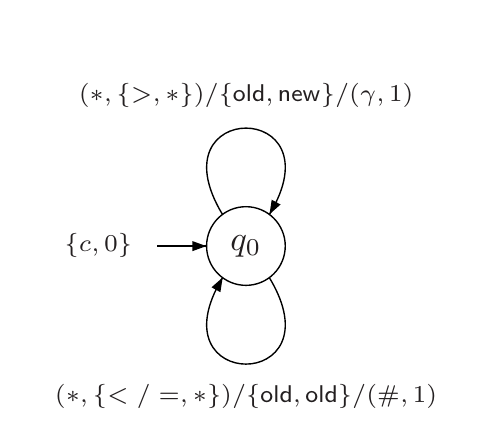}
\end{center}
\caption{$\mathcal{S}_{\filtering}$ in Example~\ref{ex:filtering}} \label{fig:filtering}
\end{subfigure}
\caption{Examples of \SRT.}
\end{figure}

\begin{example}[Sum Computation]
\label{ex:sum}
One can construct an $(\Sigma, \Gamma, \domain)$-\SRT~ to receive any possible input data value and emit the sum of the data values received thus far: 
\[
\mathcal{S}_{\summ} \defeq (\{q_0\}, q_0, 1, 0, \Delta)
\] 
where $\Sigma, \Gamma, \domain$ are all arbitrary. $\Delta$ is the set of transitions of the form $(q_0, \sigma, l, \addv, 1, \gamma, q_0)$, where $\sigma$ and $l$ are arbitrary, and $\gamma$  is a fixed output label.
Intuitively, in every position, $\mathcal{S}_{\summ}$ ignores the input label and adds the input data value into the register. Formally, if the input is an infinite sequence $\langle (*, d_i)\rangle_{i \geq 0}$, then the output is an infinite sequence $\langle (\gamma, \sum_{0 \leq i' \leq i} d_{i'})\rangle_{i \geq 0}$.
Figure~\ref{fig:sum} shows a graphical representation of the transducer.
\end{example}

\begin{example}[Random Sampling]
\label{ex:sampling}
Let $\Sigma$ be input labels, $\Gamma = \{\gamma, \#\}$ be output labels, $\domain$ be a linear group. One can construct a $(\Sigma, \Gamma, \domain)$-\SRT~ to nondeterministically output a data value from every $n$ input data values: 
\[
\mathcal{S}^n_{\sampling} \defeq (\{q_0, \dots, q_{2n-2}\}, q_0, 1, 0, \Delta)
\]
where $\Delta$ includes the following four kinds of transitions:
\begin{itemize}
\item $(q_i, \sigma, l, \oldv, 1, \#, q_{i+1})$, for any $0 \leq i \leq n-2$ or $n \leq i \leq 2n-3$
\item $(q_i, \sigma, l, \newv, 1, \gamma, q_{i+n+1})$, for any $0 \leq i \leq n-2$
\item $(q_{n-1}, \sigma, l, \newv, 1, \gamma, q_0)$
\item $(q_{2n-2}, \sigma, l, \oldv, 1, \#, q_0)$
\end{itemize}
Figure~\ref{fig:sampling} shows a graphical representation of $\mathcal{S}^2_{\sampling}$, which samples one value from every $2$ input values. 
\end{example}

\begin{example}[High-Pass Filtering]
\label{ex:filtering}
One can construct an $(\Sigma, \Gamma, \domain)$-\SRT~ to filter out all inputs whose data values are less than or equal to $c$, a constant value from the underlying data domain:
\[
\mathcal{S}_{\filtering} \defeq ( \{q_0\}, q_0, 2, \{c, 0\}, \Delta)
\] 
There are two registers, one storing the constant $c$ and the other one storing the current input data value. 
The transitions in $\Delta$ are regardless of the input label and only determined by the input data value $d$: if $d$ is greater than $c$, then update the second register with $d$ and output $(\gamma, d)$ where $\gamma$ is a fixed output label; otherwise keep the registers unchanged and output $(\#, d)$.
Figure~\ref{fig:filtering} shows a graphical representation of this transducer.
\end{example}

\subsection{Trail}

We next introduce \emph{trails}, which are essentially abstract runs of \SRT, ignoring the concrete data values of a run and recording only how the values are compared and updated.
\begin{definition}[Trail]
Let $\Sigma$ be an input label set and $\Gamma$ be an output label set, $k > 0$ is a natural number. A $(\Sigma, \Gamma, k)$-\emph{trail} is a finite word over the finite alphabet $\Sigma \times \Theta_k \times \Gamma$ where
\[ \Theta_k \defeq \{>, =, <\}^k \times \{\oldv, \newv, \addv\}^k \times \{1, \dots, k\} \] 
\end{definition}
We use three operators $\Pi_{\compare}$, $\Pi_{\update}$ and $\Pi_{\outputs}$ to extract from a trail the three components of $\Theta_k$. For example, $\Pi_{\compare}(T)$ is the sequence of comparison vectors.

An \SRT~ delineates a regular set of trails, and any run of the \SRT~ has a corresponding trail from the set. See the following definitions.
\begin{definition}[Trail Automaton]
\label{def:trail-automaton}
Let $\mathcal{S} = (Q, q_0, k, R_0, \Delta)$ be a $(\Sigma, \Gamma, \domain)$-\SRT. The \emph{trail automaton} for $\mathcal{S}$ is a nondeterministic automaton 
\[ \tilde{\mathcal{S}} \defeq \big(\Sigma \times \Theta_k \times \Gamma,~ Q \times \mathcal{W}_k,~ (q_0, W_0),~ \tilde{\Delta},~ Q \times \mathcal{W}_k \big) \] where 
\begin{itemize}
\item $\Sigma \times \Theta_k \times \Gamma$ is a finite alphabet, i.e., accepting $(\Sigma, \Gamma, k)$-trails; and
\item $Q \times \mathcal{W}_k$ is the set of states (and also accepting states), where $\mathcal{W}_k$ is the set of weak ordering of $k$ elements;
\item and $(q_0, W_0) \in Q \times \mathcal{W}_k$ is the initial state, where $W_0$ is the weak ordering between $R_0$;
\item $\tilde{\Delta} \subseteq Q \times \mathcal{W}_k \times \Sigma \times \Theta_k \times \Gamma \times Q \times \mathcal{W}_k$ is the set of transitions, which will be defined below. 
\end{itemize}
Intuitively, every state of $\tilde{\mathcal{S}}$ consists of two components $(q, W)$ where $q$ is a state of the original $\mathcal{S}$, $W$ keeps track the ordering of the values of the $k$ registers. A transition in $\tilde{\Delta}$ extends an original transition from $\Delta$ with old and new register ordering, and must satisfy two conditions:
\begin{enumerate}
\item the old register ordering is compatible with the guard of the transition;
\item the new register ordering is obtained from modifying the old register ordering with the transition's register updates.
\end{enumerate}

\end{definition}
\begin{definition}
Let $\mathcal{S}$ be an \SRT. A trail is called an $\mathcal{S}$-trail if it is accepted by the trail automaton $\tilde{S}$.
\end{definition}
Similar to runs, any prefix of an $\mathcal{S}$-trail $T$ is also an $\mathcal{S}$-trail, and we call it a prefix of $T$.

\begin{definition}
\label{def:corresponding-trail}
Let $\mathcal{S} = (Q, q_0, k, R_0, \Delta)$ be a $(\Sigma, \Gamma, \domain)$-\SRT~ and let
 \[ \rho : ~(q_0, R_0) \xrightarrow[{(\gamma_0, d'_0)}]{(\sigma_0, d_0)} (q_1, R_1) \xrightarrow[{(\gamma_1, d'_1)}]{(\sigma_1, d_1)}  \cdots \xrightarrow[{(\gamma_{n-1}, d'_{n-1})}]{(\sigma_{n-1}, d_{n-1})} (q_n, R_n) \] 
 be a run of $\mathcal{S}$. There is a corresponding $\mathcal{S}$-trail, denoted as 
\[ \tilde{\rho} \defeq \langle~ (\sigma_i, T_i, \gamma_i) ~\rangle_{i \geq 0} \] where 
every $(q_i, \sigma, T_i, \gamma, q_{i+1})$ forms the unique transition in $\Delta$ that enables the $i$-th transition step, $(q_i, R_i) \xrightarrow[{(\gamma_i, d'_i)}]{(\sigma_i, d_i)} (q_{i+1}, R_{i+1})$.
\end{definition}

There is a critical property about trails that will be used in our proofs: a trail completely determines every register's value and the output's data value, in terms of the input. We first define a family of sets of natural numbers based on a given trail.
\begin{definition}
\label{def:set}
Given a $k$-register \SRT~$\mathcal{S}$, and an $\mathcal{S}$-trail $T$. For any $0 \leq i < |T|$ and any $1 \leq j \leq k$, a set of all natural numbers $P^T_{i,j}$ can be defined as all numbers $n \geq -1$ that satisfies the following conditions:
\begin{itemize}
\item $\Pi_{\update}(T)[n][j]$ is $\newv$ or $\addv$;
\item for any $n < n' \leq i$, $\Pi_{\update}(T)[n'][j]$ is not $\newv$.
\end{itemize}
Note that $-1$ is a special value in the set, indicating that the $j$-th register is never reset by the $i$-th step.
\end{definition}
\begin{lemma}\label{thm:srt-reg}
Let $\mathcal{S}$ be an \SRT~ with $k$ registers and $\rho$ be a run of $\mathcal{S}$ over $s$ generating $t$, represented as
 \[
(q_0, R_0) \xrightarrow[{t[0]}]{s[0]} (q_1, R_1) \xrightarrow[{t[1]}]{s[1]} \cdots \xrightarrow[{t[n-1]}]{s[n-1]} (q_n, R_n)
 \]
Then for any $0 \leq i < n$ and any $1 \leq j \leq k$, 
\[
R_i[j] = \sum_{e \in P^T_{i,j}} \ite (e < 0,~ R_0[j],~ s[e])
\]
where $T$ is the corresponding trail for $\rho$ and $P^T_{i,j}$ is the set as defined in Definition~\ref{def:set}.
\end{lemma}
\begin{proof}
By induction on $i$.
\end{proof}

A trail can also be used to determine if a run it corresponds to is dead, i.e., cannot be extended with any further input.
\begin{lemma}
\label{thm:trail}
Let $\mathcal{S}$ be an \SRT. A run $\rho$ of $\mathcal{S}$ is not a prefix of any other run if and only if its corresponding trail $\tilde{\rho}$ is not a prefix of any other $\mathcal{S}$-trail.
\end{lemma}
\begin{proof}
Let $\rho$ ends at configuration $(q, R)$. If the run can be extended with a transition step $(q, R) \xrightarrow[{(\gamma, d_2)}]{(\sigma, d_1)} (q', R')$, which is enabled by a transition $(q, \sigma, T, \gamma, q')$. Then by Definition~\ref{def:corresponding-trail}, the trail $\tilde{\rho}$ can also be extended to $\tilde{\rho} \cdot (\sigma, T,\gamma)$.

In reverse, the trail automaton $\tilde{\mathcal{S}}$ accepts $\tilde{\rho}$ and let us say the ending state is $(q, W)$. By Definition~\ref{def:trail-automaton}, the ordering information between all registers has been encoded in $W$ and guarantees that whenever $(q, W)$ accepts one more trail step $(\sigma, T, \gamma)$ and transit to $(q', W')$, the guard in the transition $(q, \sigma, T, \gamma, q')$ can always be satisfied by a concrete input data value $d$. Therefore $\rho$ can also be extended.
\end{proof}

\subsection{Closure Properties}

\begin{theorem}\label{thm:srt-closure}
The \SRT-definable languages are closed under union and intersection.
\end{theorem}
\begin{proof}
Let $\mathcal{S}_1$ and $\mathcal{S}_2$ be two \SRT~ with state sets $Q_1$ and $Q_2$, as well as $k_1$ and $k_2$ registers, respectively. The union of the two transformations can be defined using a linear combination of the two \SRT. The combined \SRT~ has a state set $Q_1 \cup Q_2 \cup \{q_c\}$ and $(k_1 + k_2)$ registers (all initialized as in $\mathcal{S}_1$ and $\mathcal{S}_2$), where $q_c$ is the new initial state which simulates all possible transitions from the initial states of $\mathcal{S}_1$ or $\mathcal{S}_2$. After the first transition, the transducer reaches a state in $Q_1$ or $Q_2$ and then simulates the behavior of $\mathcal{S}_1$ or $\mathcal{S}_2$, for the rest of the input.

The intersection of the two transformations can be defined by the \SRT~ described below. The transducer has a state set $Q_1 \times Q_2 \times \{>, =, <\}^{k_1 \cdot k_2}$ and equipped with $(k_1 + k_2)$ registers. Intuitively, the transducer simulates the transitions of $\mathcal{S}_1$ and $\mathcal{S}_2$ at the same time. Every state keeps track of three things: the current state in $Q_1$, the current state in $Q_2$, and the linear order between every pair of registers from $\mathcal{S}_1$ and $\mathcal{S}_2$ using a $k_1 \times k_2$ array. Note that the initial register values are all fixed and the array can be initialized based on these values. Every transition is the combination of a transition $\delta_1$ from $\mathcal{S}_1$ and a transition $\delta_2$ from $\mathcal{S}_2$ who agree on the outputs label.
 The transition makes sure the input satisfies the guards for both $\delta_1$ and $\delta_2$, and the two output register values are equal. For example, if $\delta_1$'s output register is $i_1$ and $\delta_2$'s output register is $i_2$, then the current state's register-comparison array needs to confirm that $i_1 = i_2$, i.e., $\delta_1$ and $\delta_2$ agree on what data value to output. 
\end{proof}

\begin{theorem}
\label{thm:srt-comp}
The \SRT-definable languages are not closed under composition.
\end{theorem}
\begin{proof}
It suffices to show a counterexample to disprove the closure property. 
Now consider the \SRT~$\mathcal{S}_{\summ}$ constructed in Example~\ref{ex:sum} and a transformation combination:
\[
\mathcal{T} \defeq \semantics{\mathcal{S}_{\summ}} \cdot \semantics{\mathcal{S}_{\summ}}
\]
We now prove that $\mathcal{T}$ is not \SRT-definable. Note that $\mathcal{T}$ is functional: for any input sequence $\langle (*, d_i)\rangle_{i \geq 0}$, the output sequence is 
\[
\langle (*, \sum_{0 \leq i' \leq i} \sum_{0 \leq i'' \leq i'} d_{i''})\rangle_{i \geq 0}
\]
Let us assume the domain is integers, and every position $i$'s data value is an integer $1$. Then output at the $i$-th position is always $i \cdot (i+1) / 2$, i.e., the value grows quadratically, which is not possible in \SRT. It is not hard to argue inductively that the maximum output value at a position $i$ is the sum of past inputs $\sum_{0 \leq i' \leq i} d_{i'}$, exactly what $\mathcal{S}_{\summ}$ outputs. The contradiction concludes the proof.
\end{proof}

\subsection{Subclasses of \SRT}

In this paper, we also discuss several subclasses of \SRT, as defined below.

\begin{definition}[Deterministic \SRT]
\label{def:dsrt}
An \SRT~$\mathcal{S}$ is \emph{deterministic} (denoted as \DSRT) if for any two transitions in $\mathcal{S}$,
say $(q_1, \sigma_1, l_1, m_1, u_1, \gamma_1, q'_1)$ and $(q_2, \sigma_2, l_2, m_2, u_2, \gamma_2, q'_2)$, $q_1 \neq q_2$ or $\sigma_1 \neq \sigma_2$ or $l_1 \neq l_2$.
\end{definition}

The \SRT~$\mathcal{S}_{\summ}$ (see Example~\ref{ex:sum}) and~$\mathcal{S}_{\filtering}$ (see Example~\ref{ex:filtering}) are deterministic and also examples of \DSRT. However, \SRT~ is strictly more expressive than \DSRT. For example, $\mathcal{S}_{\sampling}$ (see Example~ref{ex:sampling}) is not a \DSRT~ because for the first two cases of transitions, if the input data value is $d$, $q_i$ can nondeterministically choose to output $(\sigma, d)$ or $(\#, d)$ (which means no-output).

\begin{definition}[Add-free \SRT]
A streaming add-free register transducer (denoted as \addfree) is an \SRT~ in which the transition set is a subset of $Q \times \Sigma \times \{>, =, <\}^{k} \times \{\oldv, \newv\}^{k} \times \Gamma \times \{1, \dots, k\} \times Q$.
\end{definition}
In other words, \addfree~ disallows addition, so registers cannot be updated by adding up the current data value read from the input. Recall the three example \SRT: $\mathcal{S}_{\sampling}$ and $\mathcal{S}_{\filtering}$ are \addfree, but $\mathcal{S}_{\summ}$ is not.

\begin{definition}[Uninitialized \SRT]
A streaming add-free register transducer is \emph{uninitialized} (denoted as \uninit) if its initial assignments $R_0$ to registers are all-zero: $(0, \dots, 0)$.
\end{definition}
For example, both $\mathcal{S}_{\sampling}$ and $\mathcal{S}_{\summ}$ are \uninit; but $\mathcal{S}_{\filtering}$ is not (the first register needs to be initialized with the constant $c$).

\begin{definition}[Dense \SRT]
A streaming add-free register transducer is \emph{dense} (denoted as \dense) its data domain $\domain$ is dense.
\end{definition}
Recall Proposition~\ref{thm:dense-discrete}, if an \SRT~ is not an \dense, its data domain must be discrete.

Above subclasses of \SRT~ are orthogonal and can be combined. For example, we write \AU~ for add-free and uninitialized \SRT, and \AD~ for add-free and dense \SRT. We now discuss these subclasses' closure properties.

\begin{theorem}
\label{thm:addfree-closure}
The \addfree-definable languages are closed under union, intersection, and composition.
\end{theorem}
\begin{proof}
As the transducers constructed in the proof of Theorem~\ref{thm:srt-closure} do not introduce any new \addv~operation, the same proof applies to \addfree.
We now prove the closure under composition.

Let $\mathcal{S}_1$ and $\mathcal{S}_2$ be two \addfree~ with state sets $Q_1$ and $Q_2$, as well as $k_1$ and $k_2$ registers, respectively. We construct an \addfree~ to define the composition $\semantics{\mathcal{S}_1} \cdot \semantics{\mathcal{S}_2}$. Similar to the \SRT~construction for intersection (see the proof of Theorem~\ref{thm:srt-closure}), the state set is $Q_1 \times Q_2 \times \{>, =, <\}^{k_1 \cdot k_2} \times \{\pos, 0, \negz\}^{k_1}$ and the number of registers is $(k_1 + k_2)$. In addition to the comparison array between registers, the state also keeps track of the register values for $\semantics{\mathcal{S}_1}$ are positive, negative or $0$. 

Each transition sequentially performs a transition $\delta_1$ from $\mathcal{S}_1$ and then a transition $\delta_2$ from $\mathcal{S}_2$. Note that in \addfree, the register values can only be updated with the input data value or $0$. Hence the data value output from $\delta_1$ can be the input data value, one of the registers' value, or $0$. The comparison arrays in the current state allow us to check the guard for $\delta_2$. The composed transition is only allowed when the output of $\delta_1$ satisfies the guard of $\delta_2$. Both the registers and the comparison arrays are updated accordingly in each composed transition. 
\end{proof}

\begin{theorem}
The \AD- and \AU- definable languages are closed under union, intersection and composition.
\end{theorem}
\begin{proof}
As the proofs of Theorems~\ref{thm:srt-closure} and~\ref{thm:addfree-closure} are agnostic to the underlying data domain and does not introduce register initialization, the same proof applies to these subclasses.
\end{proof}

\begin{theorem}
\label{thm:dsrt-closure}
The \DSRT-definable languages are closed under intersection, but not closed under union/composition.
\end{theorem}
\begin{proof}
\DSRT~ is closed under intersection because the intersection \SRT~ constructed in Theorem~\ref{thm:srt-closure} is still deterministic if both component \SRT are deterministic. \DSRT~ is not closed under composition because the counterexample for Theorem~\ref{thm:srt-comp} is deterministic.

The reason \DSRT~ is not closed under union is quite obvious. Consider two \DSRT~ that both accept the same input data word $s$ and produce different output data words $t_1$ and $t_2$. Then the union language must contain $s \otimes t_1$ and $s \times t_2$. This is not possible to be generated by a \DSRT, which only generates a unique output for an input.
\end{proof}

\begin{theorem}
The \daddfree-definable languages are closed under intersection and composition, but not closed under union.
\end{theorem}
\begin{proof}
\daddfree~ is closed under intersection because the intersection \SRT~ constructed in Theorem~\ref{thm:srt-closure} and the composition \SRT~ constructed in Theorem~\ref{thm:addfree-closure} are deterministic and add-free if their components \SRT~ are so.
The arguments in the proof of Theorem~\ref{thm:dsrt-closure} also apply to \daddfree.
\end{proof}

The critical distinction between \SRT~ and \addfree~ is that in \SRT, every register's value is always the sum of a set of data values from the input and the register initialization; while in \addfree, every register's value is just the initial value or the data value from one position of the input. 

\begin{lemma}\label{thm:addfree-reg}
Let $\mathcal{S}$ be an \addfree~ and an $\mathcal{S}$-trail $T$. Every set of natural numbers $P^T_{i,j}$ as defined in Definition~\ref{def:set} is a singleton or empty.
\end{lemma}
\begin{proof}
As the \addfree~ disallows additive updates, $P^T_{i,j}$ contains only the most recent position at which the $j$-th register is updated with \newv, if any. Otherwise, $P^T_{i,j}$ is just empty. 
\end{proof}

\section{Expressiveness}

In this section, we discuss streaming transformations that can be defined in monadic-second order logic and first-order logic. We also compare their expressiveness with \SRT~ and its subclasses.

\subsection{Logically Defined Transformations}

We now define \smso~transducers, an \mso-based model to describe streaming transformations. The model is inspired from the \mso~transducers introduced for regular transformations of finite string~\cite{Courcelle1994} or infinite strings~\cite{Alur2012}.

\begin{definition}[\smso]
A streaming MSO transducer (\smso) is a quadruple $(\Sigma, \Gamma, \domain, \phi)$ where
\begin{itemize}
\item $\Sigma$ and $\Gamma$ are the finite sets of input and output labels, respectively;
\item $\domain$ is a linear group;
\item $Q$ is a finite set of states and $q_0 \in Q$ is the initial state;
\item $\phi \in \mso(\Sigma, \Gamma, \domain)$ is a monadic second-order sentence as the transformation condition.
\end{itemize}
\end{definition}
Given an input data word $s \in \Sigma \times D$, an \smso~ guesses an output data word $t$, and checks every prefix of $s \otimes t$ against $\phi$. The formula $\phi$ needs to be satisfied by any prefix of $s \otimes t$, so that every partial output is legal based on the partial input that has been read thus far. Any \smso~$\mathcal{M}$ defines a streaming transformation:
\[\semantics{\mathcal{M}} \defeq \{ V \mid \textrm{for any } i, V[0..i] \models \phi \} \]
We call a transformation \emph{\mso-definable} if there is an \smso~ that defines the transformation. When the formula $\phi$ in the \smso~ is a first-order formula, we also call the transformation \emph{\fo-definable}.

\subsection{Inclusiveness}

We now compare the logically defined transformations with \SRT~ and its subclasses in terms of their expressiveness. We first show \mso~ is more expressive than \SRT.

\begin{theorem}[\SRT~$\subseteq$ \mso]
\label{thm:mso}
\SRT-definable transformations (and all of its subclasses) are \mso-definable.
\end{theorem}
\begin{proof}
Let $\mathcal{S} = (Q, q_0, k, R_0, \Delta)$ be an \SRT. We build an \smso $\mathcal{M}$ such that $\semantics{\mathcal{M}} = \semantics{\mathcal{S}}$. In other words, we show how to construct an \mso~formula $\phi$ such that for any transformation instance $s \otimes t$ ($s$ is the input and $t$ is the output), $s \otimes t \in \semantics{\mathcal{S}}$ if and only if $s \otimes t[0..i] \models \phi$ for any natural number $i$.

Essentially, $\phi$ guesses a run $\rho$ that produces $s \otimes t[0..i]$:
\[
(q_0, R_0) \xrightarrow[{t[0]}]{s[0]} (q_1, R_1) \xrightarrow[{t[1]}]{s[1]} \dots \xrightarrow[{t[i]}]{s[i]} (q_i, R_i)
\]
More concretely, $\phi$ first guesses the $\mathcal{S}$-trail $T$ that corresponds to the run. Note that the $\mathcal{S}$-trails can be recognized by the trail automaton $\tilde{\mathcal{S}}$ and due to the classical logic-automata connection, the guess can be done by existentially quantifying over a set of second-order variables: 
\begin{itemize}
\item for each label  $l$ in the alphabet of $\tilde{\mathcal{S}}$, a variable $X_l$ for the set of positions on which the trail is labeled $l$;
\item for each state $q$ in $\tilde{\mathcal{S}}$, a variable $X_q$ for the set of positions on which $\tilde{\mathcal{S}}$ runs to state $q$.
\end{itemize}
With the guessed trail, by Lemma~\ref{thm:srt-reg}, every register's value at every step can be expressed in terms of $s$ and the family of sets $P^T_{i,j}$. 
Moreover, every set $P^T_{i,j}$ can be defined in \mso~ according to Definition~\ref{def:set}.
Therefore $\phi$ can check, for each transition step, the three conditions described in Section~\ref{sec:config}, to make sure the step is enabled by the transition determined by the $X_l$'s and $X_q$'s.
\end{proof}

One may wonder if $\SRT$-transformations are $\fo$-definable or vice versa. The following two theorems show that the two classes of transformations are not comparable.

\begin{theorem}[\fo~$\not \subseteq$ \SRT]
Not all \fo-definable transformations are \SRT-definable.
\end{theorem}
\begin{proof}
It suffices to show an \fo-formula such that the transformation it defines cannot be defined by any \SRT. Consider the following simple \fo-formula:
\[
\phi_{\textrm{fresh}} \defeq  \forall x. \Big( \forall y. (y \prec x \rightarrow dt_{in}(y) \neq dt_{in}(x)) \leftrightarrow L_a(x)  \Big)
\]
The formula states that the transformation should give a label $a$ to the output if and only if the current input is a fresh new data value not seen before. It is easy to argue that this transformation cannot be defined by any \SRT: to determine the output label should be $a$ or not, the transducer must memoize all previous data values, which is impossible with any finite number of registers. 
\end{proof}

\begin{theorem}[\addfree~$\not \subseteq$ \fo]
Not all \addfree-definable transformations are $\fo$-definable.
\end{theorem}
\begin{proof}
Note that \addfree~ can simulate the complementation of any regular language over finite alphabet. Let $A$ be a deterministic finite automaton that recognizes a language over $\Sigma$. Then we can construct an \addfree~
\[
\mathcal{S}_{\neg A} = (\Sigma, \{T\}, \domain, Q_A, q_0^A, 0, \emptyset, \Delta_A)
\]
where $Q_A$ is the set of states for $A$ and $q_0^A$ is the initial state for $A$. $\mathcal{S}_{\neg A}$ does not have any register and the transition set $\Delta_A$ simply mimics the deterministic transitions of $A$, and outputs the unique label $T$ if the current state is \emph{not} an accepting state of $A$. Once an accepting state is reached, no transition will be available and no more input will be accepted.

Now assume \addfree~$\subseteq$ \fo, then the transformation $\semantics{\mathcal{S}_{\neg A}}$ can be defined in \fo~ as a formula $\phi_{\neg A}$ such that 
$\neg \phi_{\neg A}$ is satisfied if and only if the word is accepted by $A$.
 Nonetheless, due to the classical equivalence between finite automata and \mso, arbitrary \mso~formula $\psi$ on finite word can be converted to an automaton $A(\psi)$, which can be further converted to $\phi_{A(\psi)}$. In other words, every \mso-formula can be converted to an equivalent \fo-formula, which is obviously wrong.
\end{proof}

The non-\fo-definability also holds subclasses of \addfree, e.g., \AU, \AD, \AUD, and \daddfree, as the \addfree~ constructed above does not have any registers and is deterministic.

In conclusion, \SRT~ is strictly less expressive than \mso~ and not comparable with \fo. Precise logical characterization of \SRT~ is still unknown and we pose it as an open problem.

\section{Decision Problems}

In this section, we investigate several decision problems about \SRT. 

\subsection{Functionality}

We first consider the functionality check problem. The transformations as defined in Definition~\ref{def:transformation} are sets of input-output pairs. We call a transformation functional if it is a partial function from input data words to output data words.
\begin{definition}[Functionality]
\label{def:functionality}
A $(\Sigma, \Gamma, D)$-streaming transformation $\mathcal{T}$ is \emph{functional} if for any input data word $s \in (\Sigma \times D)^*$, the set of possible output $\{ t \mid s \otimes t \in \mathcal{T} \}$ has cardinality at most $1$. We call an \SRT~$\mathcal{S}$ \emph{functional} if $\semantics{\mathcal{S}}$ is functional.
\end{definition}

Notice that \DSRT~ are deterministic and the output, if any, is unique for a given input. Hence they are always functional. However, \SRT~ are in general nondeterministic and unnecessarily functional. 
The functionality problem is actually undecdabie for \SRT.

\begin{theorem}
\label{thm:srt-functionality}
The functionality of \SRT~ is undecidable.
\end{theorem}
\begin{proof}
We show a reduction from the halting problem of 2-counter machines~\cite{Minsky1967}. Given a 2-counter machine $\mathcal{M}$, we build an \SRT~$\mathcal{S}$ with integers as the data domain, which simulates the execution of $\mathcal{M}$. The finite control of $\mathcal{M}$ can be encoded as a finite set of states in $\mathcal{S}$. $\mathcal{S}$ is also equipped with $5$ registers: two of them are mutable and store the current values of the two counters, the other three are registers with constants $0$, $1$, and $-1$. Each step of the execution of $\mathcal{M}$ is simulated as below. For counter increment (resp. decrement), $\mathcal{S}$ makes sure the next letter has data value $1$ (resp. $-1$) by comparing with the constant register, and add the value to the corresponding register. For testing whether a counter is zero, $\mathcal{S}$ makes sure the next letter has data value is equal to the corresponding register's value and also equal to the $0$-register. Similarly, equality test between two counters is simulated by comparing the current input with two mutable register values. The whole simulation produces a fixed output for every step. Once the simulation terminates, $\mathcal{S}$ jumps to a special state which nondeterministically emits any output. Therefore, $\mathcal{M}$ halts if and only if $\mathcal{S}$ is not functional.
\end{proof}

Nonetheless, the functionality of \addfree~ (and some subclasses) is decidable. The key idea is to show that any non-functional \addfree~ can be confirmed by a bounded-size witness.

Let $\mathcal{S} = (Q, q_0, k, R_0, \Delta)$ be a $(\Sigma, \Gamma, \domain)$-\addfree. If $\mathcal{S}$ is not functional, there must be a finite string $s \in (\Sigma \times D)^*$ as witness input and two finite strings $t, p \in (\Gamma \times D)^*$ as witness outputs, such that $t \neq p$ and $\{ s \otimes t, s \otimes p \} \in \semantics{\mathcal{S}}$.
There are two runs $\rho_t$ and $\rho_p$ that produce $s \otimes t$ and $s \otimes p$, respectively. Their corresponding trails, say $T_t$ and $T_p$, must be $\mathcal{S}$-trails.
Moreover, the data values from the input $s$ should satisfy a set of constraints $\Phi(z_0, \dots, z_{n-1})$ where each $z_i$ is the data value of $s[i]$. The formula is a conjunction of equalities and inequalities that check three conditions:
\begin{enumerate}
\item every transition step in $\rho_t$ and $\rho_p$ is enabled; this part is fully determined by the trails $T_t$ and $T_p$ (see the proof of Theorem~\ref{thm:mso}); 
\item the register values at every step satisfy the ordering encoded in the current state $q^t_i$ (or $q^p_i$);
\item the two outputs $t$ and $p$ are different only at the last position. 
\end{enumerate}

\paragraph{Encoding to Strand logic.}
Recall that by Lemma~\ref{thm:addfree-reg}, every register $j$'s value can be represented as $0$ or an input data value $z_i$, whose position is determined by an \mso~formula (which computes $P^t_(i,j)$ or $P^p_(i.j)$).
Our first approach is to convert $\Phi$ to a formula in the \Strand~logic, a logic interpreted over tree-like data structures~\cite{popl11,sas11}. 
Intuitively, \Strand~formulae are interpreted over trees and are of the form $\exists \vec{x} \forall \vec{y} . \phi(\vec{x}, \vec{y})$, where $\vec{x}$ and $\vec{y}$ are groups of variables over positions, $\phi(\vec{x}, \vec{y})$ is a boolean combination of unary and binary MSO-defined predicates, as well as arithmetic constraints about data values stored in $\vec{x}$ and $\vec{y}$. 

\begin{lemma}
The functionality of \addfree~ can be reduced to the satisfiability of \Strand~logic.
\end{lemma}
\begin{proof}
First, we encode all $\mathcal{S}$-trails to a class of binary trees. For example, the trail can be encoded as the leftmost path from root to leaf. Every distinct label on a trail node is encoded as a distinct length of the rightmost path starting from this node. Then the formula $\Phi$ can be easily converted to a \Strand~formula: every MSO-defined $P^t_(i,j)$ or $P^p_(i,j)$ is allowed in \Strand, just replacing the successor predicate $S(x, y)$ with the left-child predicate $Left(x, y)$. 
\end{proof}
While the satisfiability of \Strand~ is not decidable in general, it admits several decidable fragments. Nonetheless, the constructed \Strand~formula does not belong to the syntactic decidable fragment identified in~\cite{sas11}. The reason is that the \mso~predicate for $P^t_(i,j)$ is not \emph{elastic}, because, intuitively, if a trail is contracted, the $P^t_(i,j)$ set is no longer preserved---after the contraction, a register's value for a position may be set from different positions.

In~\cite{popl11}, there is a more powerful, semantically defined decidable fragment: there is an algorithm to check if a \Strand~formula belongs to the fragment. This decision procedure may be used to solve the functionality problem when the constructed \Strand~formula falls in the fragment. However, it is still open if the functionality of \addfree~ can be reduced to the decidable fragment. We next illustrate a different way of encoding the functionality problem, leading to the decidability results.

\paragraph{Encoding to Constraint graph.}
Notice that $\Phi(z_0, \dots, z_{n-1})$ is just a conjunction of inequalities of the form 
\[z_i \sim z_j \quad \textrm{or} \quad z_i \sim c \]
 where $\sim$ is $>$, $=$ or $<$, and $c$ is a constant of the data domain $D$.
 In other words, all the constraints in $\Phi$ are distance constraints and can be converted to a constraint graph~\cite{CLRS}: every variable or constant becomes a vertex, every inequality $z_i < z_j$ becomes an edge from $z_j$ to $z_i$. The weight of the edge is $-\epsilon$ (where $\epsilon$ is a positive infinitesimal) if $\domain$ is dense or $-1$ if $\domain$ is discrete with the least positive element $1$, and every equality $z_i = z_j$ becomes a bidirected edges, both directions having weight $0$.
For constraints involving a constant, we introduce a special variable $zero$, rewrite every $z_i \sim c$ to $z_i - zero \sim c$ and add an appropriate edge between $z_i$ and $zero$. For example, $z_i < c$ becomes a an edge from $zero$ to $z_i$ with weight $c - \epsilon$.
$\Phi'$ is satisfiable if and only if there is no negative-weight cycle in the corresponding constraint graph~\cite{CLRS}.
We denote the constraint graph built from trails $T_t$ and $T_p$ as $\cg(T_t, T_p)$.

\begin{lemma}
\label{thm:graph}
An \SRT~$\mathcal{S}$ is not functional if and only if there exist two witness $\mathcal{S}$-trails $T_t$ and $T_p$ such that $\cg(T_t, T_p)$ is absent of negative-weight cycle.
\end{lemma}

Lemma~\ref{thm:graph} reduces the functionality check to the search of witness trails. The decidability is proven by showing that the length of the shortest witness trail is bounded. We first focus on the case of \AD, in which the data domain is dense.




\begin{theorem}
\label{thm:ad}
The functionality of \AD~ is {\sf NEXPTIME}.
\end{theorem}
\begin{proof}
 Let $\mathcal{S} = (Q, q_0, k, R_0, \Delta)$ be a $(\Sigma, \Gamma, \domain)$-\AD. Assume $\mathcal{S}$ is not functional and $T_t, T_p$ are the pair of shortest witness trails for the non-functionality.
 We claim that the size of $T_t$ and $T_p$ is bounded by $(|Q| \cdot \mathcal{B}_k)^2$, where $\mathcal{B}_k$ is the $k$-th ordered Bell number
  (whose growth rate is $O(2^{\textsf{poly}(k)})$). 
Then an nondeterministic algorithm can guess two trails $T_t$ and $T_p$ of length up to $\mathcal{Q}$, build $\mathcal{CG}(T_t, T_p)$ and check the absence of negative-weight cycle in polynomial time.
  
By Lemma~\ref{thm:graph}, the constraint graph $\mathcal{CG}(T_t, T_p)$ does not contain any negative-weight cycle. In addition, $T_t$ and $T_p$ are $\mathcal{S}$-trails and accepted by the trail automaton $\tilde{\mathcal{S}}$. Let $|s| =  n$ and the two runs can be represented as
\begin{align*}
\textrm{run for } T_t:~ w_0 \xrightarrow{T_t[0]} w^t_1 \xrightarrow{T_t[1]} w^t_2 \dots \xrightarrow{T_t[n-1]} w^t_n \\
\textrm{run for } T_p:~  w_0 \xrightarrow{T_p[0]} w^p_1 \xrightarrow{T_p[1]} w^p_2 \dots \xrightarrow{T_p[n-1]} w^p_n
\end{align*}
Note that $\tilde{\mathcal{S}}$ contains $|Q| \cdot \mathcal{B}_k$ states. Then if the two trails are longer than $(|Q| \cdot \mathcal{B}_k)^2$, there must be two positions $m < j$ such that both the two runs repeat their states, i.e., $w^t_m = w^t_j$ and $w^p_m = w^p_j$. Then we claim that shorter witness trails can be constructed from the following two contracted trails:
\begin{align*}
T'_t ~\defeq~ T_t[0..m-1] \cdot T_t[j..n-1] \\
T'_p ~\defeq~ T_p[0..m-1] \cdot T_p[j..n-1]
\end{align*}

By Lemma~\ref{thm:graph}, it remains to show the absence of negative-weight cycle in $\cg(T'_t, T'_p)$.
 We split the vertices of the constraint graph into two groups: $A$ for variables $z_0$ through $z_{m-1}$ as well as the special variable $zero$; $B$ for variables $z_j$ through $z_{n-1}$. There are three kinds of edges: 
 \begin{itemize}
 \item[$E_A$]: edges within $A$.
  \item[$E_B$]: edges within $B$.
   \item[$E_{AB}$]: edges between $A$ and $B$.
 \end{itemize}
 Intuitively, $E_A$ is the constraints for the first $m$ steps of transitions; $E_B$ is the constraints for the last $(n-j)$ steps of transitions in which the register values are also defined in the last $(n-j)$ steps; $E_{AB}$ is the constraints or the last $(n-j)$ steps of transitions in which the register values are defined in the first $m$ steps.
 
 Notice that the $E_A \cup E_B$ portion of $\cg(T'_t, T'_p)$ is isomorphic to the corresponding portion of $\cg(T_t, T_p)$ and does not contain any negative-weight cycle (otherwise the same cycle already exists in $\cg(T_t, T_p)$ and $T_t$ and $T_p$ cannot be witnesses); in other words, any negative weight cycle of $\cg(T'_t, T'_p)$ must involve at least two edges from $E_{AB}$, one from $A$ to $B$ and another one from $B$ to $A$. 
  Let the edge from $A$ to $B$ be $z_a \rightarrow z_b$ and the edge from $B$ to $A$ be $z_c \rightarrow z_d$, where $a < b$ and $d < c$ ($a$ and $d$ are not necessarily different; so are $b$ and $c$). 
  We next consider two cases: the cycle involves the vertex $zero$ or not.
 
 \paragraph{Case 1: the cycle does not involve zero.} Based on the construction of the constraint graph, $z_a$ and $z_d$ are equivalent to two register values that coexist after the first $m$ steps. Now we consider two cases. First, if $a = d$, then $z_b$ and $z_c$ compare to the same register's value, and the same cycle should already exist in $\mathcal{CG}(T_t, T_p)$, contradiction.
Second, if $a \neq d$, the cycle is split into two segments: 
\begin{itemize}
\item[$S_1$]: $z_a \rightarrow z_b \rightarrow \dots \rightarrow z_c \rightarrow z_d$
\item[$S_2$]: $z_d \rightarrow \dots \rightarrow z_a$
\end{itemize}
As the cycle does not involve $zero$ and all edges are of non-positive weight, both $Weight(S_1)$ and $Weight(S_2)$ are non-positive and at least one of them is negative. As $z_a$ and $z_d$ represent two coexisting register values, the repeated program state $w^t_m$ have already imposed an edge from $z_a \rightarrow z_d$ with a non-positive weight (negative or $0$, depending on the $Weight(S_2)$). The edge $z_a \rightarrow z_d$ can serve as a shortcut for $S_2$ and be combined with $S_1$ to form a negative-weight cycle within $E_A$ only, contradiction.

\paragraph{Case 2: the cycle involves zero.} As $zero$ belongs to group $A$, the cycle is split to three segments: 
\begin{itemize}
\item [$S_1$]: $z_b \rightarrow \dots \rightarrow z_c$
\item [$S_2$]: $z_c \rightarrow \dots \xrightarrow{c_1} zero$
\item [$S_3$]: $zero \xrightarrow{c_2} \dots \rightarrow z_b$
\end{itemize}
 Recall that when $\domain$ is dense, hence only the two edges connecting $zero$ have weights $c_1$ and $c_2$, the only two values that can be not $0$ or $-\epsilon$. Therefore we have $Weight(S_1) \leq 0$, $Weight(S_2) \leq c_1$, $Weight(S_3) \leq c_2$,  $c_1 + c_2 \leq 0$ and at least one of these inequalities is strict.
and assume $z_a$ is the vertex within the loop with the greatest index $a$. Then the loop consists of two segments: 
Based on the semantics of the constraint graph, $z_b$ and $z_c$ are comparing with two register values $r_b$ and $r_c$ that coexist after the first $m$ steps. From $S_2$ we have $z_c \geq r_c \geq -c_1$; from $S_3$ we have ``$c_2 \geq r_b \geq z_b$. As $c_1 + c_2 \leq 0$, the chain of inequalities guarantees $r_c > r_b$ or $r_c = r_b$, and the ordering must have been encoded to the state $w^t_m$. In both cases, the state cannot lead to the next $(n-j)$ steps of transitions and form the chain $S_1$, which requires $r_b \geq r_c$ or $r_b > r_c$.
\end{proof}

\begin{theorem}
\label{thm:au}
The functionality of \AU~ is {\sf NEXPTIME}.
\end{theorem}
\begin{proof}
The proof is similar to Theorem~\ref{thm:ad}. The same arguments can be made for Case 1 of the proof. For Case 2, we can still split the cycle into three segments $S_1$, $S_2$ and $S_3$. As the initial register values are all $0$ in \AU, i.e., $c_1 = c_2 = 0$, the inequalities we have can be simplified to $Weight(S_1) \leq 0$, $Weight(S_2) \leq 0$, $Weight(S_3) \leq 0$, and at least one of these inequalities is strict. Then similarly we can argue the existence of two register values $r_b$ and $r_c$ that coexist after the first $m$ steps of transitions such that $r_b > r_c$ or $r_b = r_c$. In both cases, a contradiction can be found or a negative weight cycle already exists within $E_A$.
\end{proof}

\begin{theorem}
\label{thm:addfree}
The functionality of \addfree~ is {\sf 2NEXPTIME}.
\end{theorem}
\begin{proof}
Given Theorems~\ref{thm:ad} and~\ref{thm:au} proved, what remains is to give an algorithm to check the functionality of \SRT whose data domain $\domain$ is discrete and whose registers are initialized. Again, the previous arguments for Case 1 still works and let us focus on Case 2.

We still split the cycle into $S_1$, $S_2$ and $S_3$. The incoming and outgoing edges for $zero$ are still of weight $c_1$ and $c_2$, respectively. Now as $\domain$ is discrete, all other edges are of weight $0$ or $-1$. Hence $c_1 + c_2 < L - 2 < \mathcal{Q} - 1$ where $L$ is the length of the cycle and $\mathcal{Q}$ is the number of states of the automaton we built in the proof of Theorem~\ref{thm:ad}. Notice that the two edges connecting $zero$ mean $c_1$ and $-c_2$ are two initial register values. In other words, the size of the interval $[-c_2, c_1]$ is bounded by $\mathcal{Q}-1$. Then for each register, we can introduce $\mathcal{Q}$ states to keep track of the value of the register, one state for each value in $[-c_2, c_1]$ and one extra state for values not in the range. Extending the original automaton with these states results in an automaton with $\mathcal{Q}^{k+1}$ states. We can repeat all previous arguments and claim that the shortest cycle involving initial values $-c_2$ and $c_1$ is bounded by length $\mathcal{Q}^{k+1}$.

This procedure can be continued for $k-1$ times, each time removing an interval between two initial register values at the cost of exponential state blowup. The final bound we get $\mathcal{Q}^{{k+1}^{k-1}}$, which is double exponential to the size of the original transducer $\mathcal{S}$.
\end{proof}

\begin{corollary}
For a fixed $k$, the functionality of \addfree$(k)$ is {\sf NP}.
\end{corollary}
\begin{proof}
The exponential blowup in previous three theorems is for $k$, the number of registers.
\end{proof}

\subsection{Reactivity and Inclusion}

A desirable property of \SRT~ is the reactivity, which intuitively means the transducer never stuck, i.e., it can take arbitrary infinite stream of input and generate infinite output stream.

\begin{definition}[Reactivity]
A $(\Sigma, \Gamma, D)$-streaming transformation $\mathcal{T}$ is \emph{reactive} if for any input data word $s \in (\Sigma \times D)^*$, there exists an output $d \in (\Gamma \times D)^*$ such that $s \otimes t \in \mathcal{T}$. We call an \SRT~$\mathcal{S}$ \emph{reactive} if $\semantics{\mathcal{S}}$ is reactive.
\end{definition}

We also consider the inclusion problem, i.e., if an \SRT's transformations are always transformations of another \SRT:
\begin{definition}[Inclusion]
Given two $(\Sigma, \Gamma, D)$-\SRT~$\mathcal{S}$ and $\mathcal{S'}$, we say $\mathcal{S}$ is included in $\mathcal{S'}$ if $\semantics{\mathcal{S}} \subseteq \semantics{\mathcal{S'}}$.
\end{definition}

When \SRT~ is deterministic, the reactivity problem can be reduced to the inclusion problem.
\begin{lemma}
\label{thm:reactivity-inclusion}
The reactivity problem of \DSRT~(and also \daddfree) reduces to the inclusion problem for the corresponding inclusion problem.
\end{lemma}
\begin{proof}
Given an \SRT~$\mathcal{S}$, we can extend $\mathcal{S}$ to a new \SRT~$\mathcal{S'}$, which has one more special output label $\gamma_{\bot}$, one more special state $q_{\bot}$. $\mathcal{S}$ also adds two kinds of extra transitions: 1) for all guard conditions that do not have any available transitions in $\mathcal{S}$, add a transitions that switches to $q_{\bot}$; 2) from $q_{\bot}$ there is only a loop transition that emits $(\gamma_{\bot}, 0)$ forever. In summary, $\mathcal{S'}$ mimics $\mathcal{S}$ on all inputs as long as there is an output. Whenever $\mathcal{S}$ is stuck, $\mathcal{S'}$ continues and consistently emits the dumb output $(\gamma_{\bot}, 0)$. Obviously, $\mathcal{S}$ is included in $\mathcal{S'}$. Moreover, when $\mathcal{S}$ is deterministic, $\mathcal{S'}$ emits the dumb output only if the input is not accepted by $\mathcal{S}$. In other words, $\mathcal{S'}$ is also included in $\mathcal{S}$ if and only if $\mathcal{S}$ is reactive.
\end{proof}

Therefore we discuss the two problems together for \DSRT~ and \daddfree.


\begin{theorem}
The reactivity problem of \DSRT~ is undecidable.
\end{theorem}
\begin{proof}
The proof is similar to that of Theorem~\ref{thm:srt-functionality}. A similar \SRT~$\mathcal{S}$ can be constructed from a given 2-counter machine $\mathcal{M}$. The difference is that the \SRT~ here does not emit arbitrary output after the simulation. In stead, when the input is not as expected in the simulation, i.e., no transition is available, $\mathcal{S}$ jumps to a special state, which repeatedly emits a dumb output such as $(\gamma, 0)$, regardless of the input. Therefore, $\mathcal{S}$ terminates if and only if the simulated execution of $\mathcal{M}$ terminates. Notice that the constructed $\mathcal{S}$ here is deterministic, hence the halting problem of 2-counter machines is reduced to the reactivity of \DSRT.
\end{proof}

\begin{corollary}
The reactivity problem and inclusion problem of \SRT~ are both undecidable.
\end{corollary}
\begin{proof}
By Lemma~\ref{thm:reactivity-inclusion}.
\end{proof}

\begin{theorem}
\label{thm:daddfree}
The inclusion problem (and also the reactivity problem) of \daddfree~ is {\sf 2NEXPTIME}.
\end{theorem}
\begin{proof}
To show an \daddfree~$\mathcal{S}$ is not included in another \daddfree~ $\mathcal{S'}$, it suffices to show an input data word $s$, over which there are two deterministic run: $\rho$ for $\mathcal{S}$ and $\rho'$ for $\mathcal{S'}$.
The two runs end at configurations $(q, R)$ and $(q', R')$, respectively, such that $(q, R)$ still has possible transition steps from it, but there is no more transition step is available from $(q', R')$. In other words, $\rho$ is a prefix of another run but $\rho'$ is not a prefix of any other run.
By Lemma~\ref{thm:trail}, the corresponding trails, $\tilde{\rho}$ is also a prefix of another $\mathcal{S}$-trail, but $\tilde{\rho'}$ is not a prefix of any other $\mathcal{S}$-trail. Let the two runs of the trails be of the form
\begin{align*}
\textrm{run for } \tilde{\rho}:~ w_0 \xrightarrow{T[0]} w_1 \xrightarrow{T[1]} w_2 \dots \xrightarrow{T[n-1]} w_n \\
\textrm{run for } \tilde{\rho'}:~ w'_0 \xrightarrow{T'[0]} w'_1 \xrightarrow{T'[1]} w'_2 \dots \xrightarrow{T'[n-1]} w'_n
\end{align*}
Similar to the construction of constraints we constructed for checking functionality (see the beginning of this section), we can build a set of constraints as the desired properties of the $n$ input data values. They guarantee the two runs $\rho$ and $\rho'$ of length $n$ can be constructed and satisfy the desired properties: $\rho$ can be extended with some input and $\rho'$ cannot be extended with any further input. Moreover, these constraints can be similarly solved by building a constraint graph and checking the absence of negative-wight cycle (see Lemma~\ref{thm:graph}).

Now checking inclusion is reduced to the search of two witness trails. Similar to the proof of Theorem~\ref{thm:ad}, we can show that the length of the shortest witness trail, if any, is bounded by $(|Q| \cdot \mathcal{B}_k)^2$ if the data domain is dense. 
Then similar to the proof of Theorem~\ref{thm:addfree}, we can use the same techniques to handle \daddfree~ in general, with an exponential blowup. Therefore the algorithm has the same complexity {\sf 2NEXPTIME}.

By Lemma~\ref{thm:reactivity-inclusion}, the reactivity problem can be reduced to the inclusion problem in polynomial time and also solved in {\sf 2NEXPTIME}.
\end{proof}

\paragraph{Remark:} The determinism is a critical assumption for Theorem~\ref{thm:daddfree}. Otherwise, given an input data word, there are many possible runs and trails and the no-extension property cannot be determined by checking a single or a fix number of witness trails. We leave the reactivity and inclusion problems of \addfree~ (and also its nondeterministic subclasses \AD~ and \AU) as open problems for future work.

\section{Conclusion}

We propose streaming register transducer as a natural machine model for implementations of transformations of infinite ordered-data words. This model assumes a linear group as the underlying data domain whose values are stored in a fixed number of registers. It supports nondeterministic transitions with linear-order comparison between and additive updates to registers using the input data value.
We investigate several subclasses of \SRT: the transitions are deterministic, the additive updates are disallowed, the registers are uninitialized, or the data domain is dense.
We show \SRT~ and its subclasses are strictly less expressive than \mso~ and not comparable with \fo.
We also investigate several decision problems of \SRT, including functionality, reactivity and inclusion. We prove the undecidability of these problems for \SRT. We also prove the functionality for add-free \SRT~ and the reactivity and inclusion for deterministic add-free \SRT~ are decidable. We leave precise logical characterization of \SRT~ and decidability of reactivity/inclusion for nondeterministic add-free \SRT~ as open problems.


\bibliography{refs}


\end{document}